\documentclass[%
 reprint,
 amsmath,amssymb,
 aps,
 pra,
longbibliography
]{revtex4-1}
\usepackage{graphicx}% Include figure files
\usepackage{dcolumn}% Align table columns on decimal point
\usepackage{bm}% bold math
\usepackage{braket}
\usepackage{amsmath}
\usepackage{bbold}
\usepackage{xcolor}
\usepackage{amsthm}
\usepackage{physics}
\newtheorem{theorem}{Theorem}
\newtheorem{proposition}{Proposition}
\newtheorem{corollary}{Corollary}[theorem]

\theoremstyle{definition}

\theoremstyle{remark}

\begin{document}

\preprint{APS/123-QED}

\title{The tight Second Law inequality for coherent quantum systems \\
and finite-size heat baths}

\author{Marcin {\L}obejko}
\affiliation{Intitute of Theoretical Physics and Astrophysics, University of Gdansk, Gda\'nsk, Poland}
%\author{Micha{\l} Horodecki}%
%\affiliation{Intitute of Theoretical Physics and Astrophysics, University of Gda\'nsk, Gda\'nsk, Poland}


\date{\today}% It is always \today, today,

\begin{abstract}
We propose a new form of the Second Law inequality that defines a tight bound for extractable work from the non-equilibrium quantum state. In classical thermodynamics, the optimal work is given by the difference of free energy, what according to the result of Skrzypczyk \emph{et al.} can be generalized for individual quantum systems. The saturation of this bound, however, requires an infinite bath and an ideal energy storage that is able to extract work from coherences. The new inequality, defined in terms of the ergotropy (rather than free energy), incorporates both of those important microscopic effects. In particular, we derive a formula for the locked energy in coherences, i.e. a quantum contribution that cannot be extracted as a work, and we find out its thermodynamic limit. Furthermore, we establish a general relation between ergotropy and free energy of the arbitrary quantum system coupled to the heat bath, what reveals that the latter is indeed the ultimate thermodynamic bound regarding work extraction, and shows that ergotropy can be interpreted as the generalization of the free energy for the finite-size heat baths.
\end{abstract}

\maketitle
The quantum thermodynamics is an emerging theory with the main goal to generalize the Laws of Thermodynamics, valid in a macroscopic domain, to the energetic description of individual quantum systems. 
%Alongside many other approaches, recently, quantum thermodynamics has been formulated as a resource theory of non-equilibrium quantum states \cite{Rio2011, Horodecki2013, Aberg2013, Brandao2015, Skrzypczyk2014, Guryanova2016}. 
Alongside many other approaches, recently, quantum thermodynamics has been formulated as a general unitary dynamics 
%of the closed thermal machine 
and a resource theory of non-equilibrium quantum states \cite{Rio2011, Horodecki2013, Aberg2013, Brandao2013, Skrzypczyk2014, Brandao2015, Lostaglio2015, Lostaglio2015a, Guryanova2016}.
The basic question to answer in this framework is: how much work can be extracted providing a particular resource? %Here, however, we are facing the generally unsolved problem what is really the work in the quantum domain?
%However, in order to answer this, firstly one should define what is really the work in a quantum domain. 
%However, in the first place one should define what is really the work in a quantum domain. 

In order to answer this question, in the first place, one should define what is really the work in a quantum domain, since different definitions vary from one to another framework and regimes of interests \cite{Alicki1979, Yukava2000, Allahverdyan2004work, Talkner2007, Talkner2016, Jarzynski2015, Binder2015, Horodecki2013,Brandao2015,Skrzypczyk2014, Aberg2013, Frenzel2014, Gallego2016, Hayashi2017, Sampaio2018}. The lack of consensus in this field is mainly due to the presence of coherences in quantum states \cite{Perarnau2017, Korzekwa2016, Lostaglio2015, Cwiklinski2015, Aberg2014, Narasimhachar2015, Uzdin2015} and the appearance of work fluctuations \cite{Bochkov1977,Jarzynski1997, Crooks1999, Campisi2011}. For autonomous thermal machines, where the work reservoir is explicit, one of the most promising concept is the translationally invariant energy storage (with dynamics equivalent to the physical weight), where a change of its average energy corresponds to the work \cite{Skrzypczyk2014}.
%Nevertheless, recently it has been proposed a very promising idea of the translationally invariant energy storage (with dynamics equivalent to the physical weight) where the change of its average energy corresponds to the work \cite{Skrzypczyk2014}. 

It was shown that the work reservoir given by the weight is consistent with fluctuations theorems \cite{Alhambra2016, Aberg2018, Richens2016}, likewise, it can be used to derive the Third Law of Thermodynamics \cite{Masanes2017} or to an analysis of the optimal performance of heat engines \cite{Brunner2012, Lobejko2020}. In particular, according to the first paper introducing the weight idea \cite{Skrzypczyk2014}, Skrzypczyk \emph{at al.} proved that the optimal extracted work $W$ from a quantum state $\hat \rho_S$, in a contact with a thermal reservoir at temperature $T = \beta^{-1}$, is bounded by the difference of its non-equilibrium free energy:
\begin{equation} \label{second_law}
    W \le F(\hat \rho_S) - F(\hat \tau_S),
\end{equation}
where $F(\hat \rho) = E(\hat \rho) - T S(\hat \rho)$ and $E(\hat \rho)$ is an average energy, $S(\hat \rho)$ is the von Neumann entropy and $\hat \tau_S = \mathcal{Z}_S^{-1} e^{-\beta \hat H_S}$ is a Gibbs state according to the free Hamiltonian $\hat H_S$ with partition function $\mathcal{Z}_S$.

Inequality \eqref{second_law} encapsulates the quantum form of the Second Law of Thermodynamics, what especially restricts all possible micro engines to operate below the universal Carnot efficiency. However, as presented by the authors, the optimal work $W_{max} = F(\hat \rho_S) - F(\hat \tau_S)$ is only attainable under two strong conditions: (i) the requirement of an infinite heat bath and (ii) the average energy conservation. %The first assumption is required in order to split the whole protocol into an infinite number of steps such that the dissipation in each of them can be made arbitrary small and then the optimal work can be extracted in a quasi-reversible process. 
The first assumption is required to split the protocol into an infinite number of steps, such that the optimal work can be extracted in a quasi-reversible process.
However, it is seen that saturation of the Second Law is never possible for the physical (i.e. finite) heat baths. On the other hand, the second assumption is imposed in order to make possible the full work extraction from coherences, but in this approach the First Law is not independent of the initial state. In contrary, imposing the strict form of the energy-conservation once again (in general) makes the Second Law not tight \cite{Korzekwa2016}. In other words, as long as the inequality \eqref{second_law} provides the universal upper bound for the work extraction it does not answer the question when it can be saturated, what requires additional information about microscopic details of the heat bath and state of the work reservoir. %Technically speaking, the right hand side of the inequality \eqref{second_law} is independent of the heat bath free Hamiltonian $\hat H_B$ and the initial state of the weight $\hat \rho_W$. 

\subsection*{Tight Second Law}
In this work, we derive a generalized formula which incorporates both of those microscopic effects and makes the inequality tight. We
reveal that for arbitrary protocol $\hat \rho_S \otimes \hat \tau_B \otimes \hat \rho_W \to \hat U \hat \rho_S \otimes \hat \tau_B \otimes \hat \rho_W \hat U^\dag$, where $\hat U$ is the energy-conserving and translationally invariant unitary \cite{Alhambra2016}, the tight Second Law can be written in the form:
\begin{equation} \label{tight_second_law}
%W \le R(\hat \sigma_S \otimes \hat \tau_B, \hat H_S + \hat H_B)    
W \le R(\hat \sigma_S \otimes \hat \tau_B)    
\end{equation}
where $\hat \tau_B = \mathcal{Z}_B^{-1} e^{-\beta \hat H_B}$ is the Gibbs state of the bath and $R(\hat \rho)$ is the ergotropy of the state $\hat \rho$ \cite{Allahverdyan2004}:
\begin{equation}
R(\hat \rho) := \max_{\hat U - \text{unitary}} \Tr[(\hat H - \hat U^\dag \hat H \hat U) \hat \rho], 
\end{equation}
i.e. the maximal energy extracted from an arbitrary unitary channel $\hat \rho \to \hat U \hat \rho \hat U^\dag$ with fixed Hamiltonian $\hat H$ (according to the quantity $R(\hat \sigma_S \otimes \hat \tau_B)$ the constant Hamiltonian is equal to $\hat H_S + \hat H_B$). Further, $\hat \sigma_S$ is the so-called \emph{control-marginal state} \cite{Lobejko2020}, which for initial product states, i.e. $\hat \rho_{SW} = \hat \rho_S \otimes \hat \rho_W$, 
can be represented as:
\begin{equation} \label{sigma}
    %\hat \sigma_S = \Tr[\hat S \hat \rho_S \otimes \hat \rho_W \hat S^\dag],
    \hat \sigma_S = \int dt \ p(t) \ \hat U_t \hat \rho_S \hat U_t^\dag
\end{equation}
where %$p(t) = \bra t \hat \rho_W \ket t$, 
$p(t) = \Tr[\hat \rho_W \dyad{t}_W]$, 
$\hat U_t = e^{-i \hat H_S t}$ and $\ket t_W$ is a canonically conjugate `time state' with respect to the energy states $\ket \varepsilon_W$. Despite that in this paper we only consider product states, nevertheless, we stress that definition of the control-marginal operator $\hat \sigma_S$ can be generalized for arbitrary correlated state $\hat \rho_{SW}$, such that inequality \eqref{tight_second_law} is still valid (see Appendix). 

Inequality \eqref{tight_second_law} reveals that the optimal work done on a weight via energy-conserving unitary dynamics is equal to the ergotropy of the composite state $\hat \sigma_S \otimes \hat \tau_B$. On the other hand, a concept of the ergotropy as the maximal extractable work arises from the cyclic non-autonomous protocols of closed quantum systems (with implicit work reservoirs) \cite{Allahverdyan2004}, what is intensively studied area of so-called `quantum batteries' \cite{Alicki2013, Hovhannisyan2013, Giorgi2015, Binder2015, Perarnau2015, Campaioli2017}. This proves the important connection between those two frameworks, however, at the same time, it emphasizes the fundamental difference between them, namely, a replacement of the marginal state $\hat \rho_S$ to the control-marginal state $\hat \sigma_S$. As it is presented below, this change significantly affects the work extraction from quantum coherences and shows that the non-autonomous framework is just a special example of the weight dynamics. 

%As was mentioned previously, the tight inequality \eqref{tight_second_law} involves two important microscopic effects: an ability of extraction work from coherences and finite-size heat bath effects. Now, we would like to discuss both of those separately. 

%\emph{Work extraction from coherences.} %Let us start with quantum coherences. 
\subsection*{Work extraction from coherences}
%One of the most important difference between classical and quantum thermodynamics is that quantum systems are able to perform the work via coherences. However, we reveal below that it is only possible if the work reservoir has coherences as well. In other words, for the quantum process of work extraction it is crucial that the weight is the energy reservoir likewise that it is the reservoir of coherences.
Firstly, one should notice that the optimal work $W_{max} = R(\hat \sigma_S \otimes \hat \tau_B)$ depends implicitly on the state of the weight $\hat \rho_W$ through the control-marginal state $\hat \sigma_S$. We point out, however, that the channel \eqref{sigma} only affects the off-diagonal elements of the density matrix $\hat \rho_S$, i.e. if $[\hat \rho_S, \hat H_S] = 0$ then $\hat \sigma_S = \hat \rho_S$. From this follows that for quasi-classical diagonal states the optimal work extraction protocol is independent of the weight state at all. Nevertheless, if we consider a coherent state $\hat \rho_S$ it is no longer true, i.e. in general $\hat \sigma_S \neq \hat \rho_S$, and the optimal value of work $W_{max}$ indirectly depends on the state of the weight (and especially of its amount of coherences). Since the channel \eqref{sigma} is a mixture of unitaries we can define a non-negative quantity:
\begin{equation} \label{locked_ergotropy}
    \Delta_L(\hat \rho_S, \hat \rho_W, \hat \tau_B) := R(\hat \rho_S \otimes \hat \tau_B) - R(\hat \sigma_S \otimes \hat \tau_B) \ge 0,
    %\Delta_L := R(\hat \rho_S \otimes \hat \tau_B) - R(\hat \sigma_S \otimes \hat \tau_B) \ge 0,
    %W_L := R(\hat \rho_S \otimes \hat \tau_B) - R(\hat \sigma_S \otimes \hat \tau_B) \ge 0,
    %\Delta_L := R(\hat \rho_S \otimes \hat \tau_B) - R(\hat \sigma_S \otimes \hat \tau_B) \ge 0,
\end{equation}
and then the Second Law can be expressed in the form: 
\begin{equation} \label{tight_with_locked}
    W \le R(\hat \rho_S \otimes \hat \tau_B) - \Delta_L(\hat \rho_S, \hat \rho_W, \hat \tau_B).
\end{equation}
Here, we call $\Delta_L(\hat \rho_S, \hat \rho_W, \hat \tau_B)$ a \emph{locked energy}, i.e. a quantum thermodynamic resource that is bounded in coherences and cannot be extracted as a work via a state of the weight $\hat \rho_W$. 

In this way, we can introduce a concept of the \emph{ideal weight}, i.e. an energy storage system that is able to the full work extraction from coherences with $\Delta_L = 0$. In particular, this is the case if the state of the weight tends to the time state, i.e. $\hat \rho_W \to \dyad{t}_W$, such that we have $\hat \sigma_S \to  \hat U_t \hat \rho_S \hat U_t^\dag$ and $\Delta_L \to 0$. The time state of the weight is an extreme and idealized example of the system with an `infinite amount of coherence' (see e.g. \cite{Korzekwa2016, Aberg2013}), and in this sense it is able to perform a unitary transformation on the subsystem and achieve the optimal work extraction. In the opposite limit, where the work storage tends to the energy eigenstate, i.e. $\hat \rho_W \to \dyad{\varepsilon}_W$, the control-marginal state loses all of the coherences, such that $\hat \sigma_S \to D[\hat \rho_S]$ (where $D[\cdot]$ is a dephasing channel in the energy basis), and hence $\Delta_L$ is maximal.

The fact that $\hat \sigma_S = D[\hat \rho_S]$ for incoherent states of the work reservoirs was previously observed and called `work-locking' \cite{Korzekwa2016}. In this research, authors discuss only the diagonal states of the energy storage and the work extraction from coherences was analyzed via additional ancillary system, a source of coherence, acting as a catalyst. Here, in contrary, we allow to use the coherent states of the weight (i.e. a fully quantum energy storage) and reveal how this can `unlocked' the extracted work. %In this case, the state $\hat \rho_S$ corresponds to the total available resource (including any ancillary system) and $\Delta_L$ quantifies the total locked energy. 

%\emph{Ergotropy vs. Free energy.} 
\subsection*{Ergotropy vs. Free energy}
Inequality \eqref{tight_with_locked} expresses the tight form of the Second Law for coherent quantum systems and finite-size heat baths, which separately includes the locked energy $\Delta_L(\hat \rho_S, \hat \rho_W, \hat \tau_B)$ and optimal work replaced by the ergotropy (instead of free energy) $R(\hat \rho_S \otimes \hat \tau_B)$. Those quantities depend on the heat bath equilibrium state $\hat \tau_B$ that is defined both by the temperature $T$ and the Hamiltonian $\hat H_B$. In contrary, the only information coming from the heat bath included in the Second Law given by Eq. \eqref{second_law}, formulated solely in terms of the free energy, is the temperature $T$. This ignorance of the microscopic details of the heat bath as a consequence makes the inequality in general not tight. 

Now, we would like to state a general relation between ergotropy and free energy for quantum systems coupled to the heat bath. We independently prove that for arbitrary quantum state $\hat \rho_S$ with Hamiltonian $\hat H_S$ and arbitrary Gibbs state $\hat \tau_B$ with Hamiltonian $\hat H_B$ the following inequality holds:
\begin{equation} \label{ergotropy_vs_free_energy}
    R(\hat \rho_S \otimes \hat \tau_B) \le F(\hat \rho_S) - F(\hat \tau_S),
\end{equation}
where $\hat \tau_S$ is the Gibbs state according to the Hamiltonian $\hat H_S$. %and temperature $T$, where latter is inherited from the bath. 
In this formula, both the Gibbs states and free energy are defined with respect to the same and arbitrary temperature $T$. 

As it is seen, the right hand side of the inequality  \eqref{ergotropy_vs_free_energy} does not depend on the Hamiltonian $\hat H_B$, what reveals that free energy is indeed the ultimate thermodynamic bound valid for all possible heat baths (see Eq. \eqref{second_law}). %The interesting question is what are models of the Hamiltonian $\hat H_B$ that lead to the saturation of this inequality?  %We refer to this class of the environments, achieved in the thermodynamic limit, as the \emph{generic heat baths}, and we simply define it by an equality of the formula \eqref{ergotropy_vs_free_energy}. 
According to the result of Skrzypczyk \textit{et. al.} \cite{Skrzypczyk2014}, the bound can be reached for infinite heat baths, i.e. in the thermodynamic limit. This relation proves the Second Law of Thermodynamics for any framework with the optimal extracted work identified as the ergotropy (e.g. for the unitary transformations of states), and shows that ergotropy on it its own is a generalization of free energy for the finite-size heat baths.

Further, by this formula, we are able to derive the thermodynamic limit of the locked energy $\Delta_L$. Basically, if for the heat bath holds $R(\hat \rho_S \otimes \hat \tau_B) = F(\hat \rho_S) - F(\hat \tau_S)$ for arbitrary density matrix $\hat \rho_S$, then 
\begin{equation} \label{locked_ergotropy_bound}
     \Delta_L(\hat \rho_S, \hat \rho_W, \hat \tau_B) = T \left[ S(\hat \rho_S) - S(\hat \sigma_S) \right],
\end{equation}
i.e. the locked energy in coherences is equal to the difference of entropy between the state $\hat \rho_S$ and control-marginal state $\hat \sigma_S$ multiplied by the bath temperature $T$. Equation \eqref{locked_ergotropy_bound} gives us an interesting formula how the quantum energy storage is able to extract work from coherences of a system in a contact with the macroscopic heat bath. However, we would like to emphasize that this formula is not the upper bound of the locked energy $\Delta_L$ (as the free energy was for the optimal work), but rather it is the thermodynamic limit. In the next paragraph we provide a numerical simulation of a particular example where the locked energy for finite-size bath can be bigger then value given by Eq. \eqref{locked_ergotropy_bound}, and moreover it can be even non-monotonic with respect to the growing size of the heat bath. 

%One should additionally notice that similarly to Eq. \eqref{ergotropy_vs_free_energy}, inequality in the form $R(\hat \rho) \le F(\hat \rho) - F(\hat \tau)$ was previously proven \cite{Allahverdyan2004}. However, in this case the temperature $T$, appearing in the definition of the Gibbs state $\hat \tau$ and free energy, is fixed, and defined by the constant entropy relation: $S(\hat \rho) = S(\hat \tau$). In contrary, in the Eq. \eqref{ergotropy_vs_free_energy} the bath $\hat \tau_B$ is explicitly introduced (with arbitrary temperature $T$) as the ancillary system with respect to the proper thermodynamic resource given by the non-equilibrium state $\hat \rho_S$. This refers to the the common physical situation of the quantum system coupled to the heat bath. As a consequence, free energy is defined only with respect to the system $S$ and with temperature inherited from the bath. 

\begin{figure}[t]
    \centering
    \includegraphics[width = 0.45\textwidth] {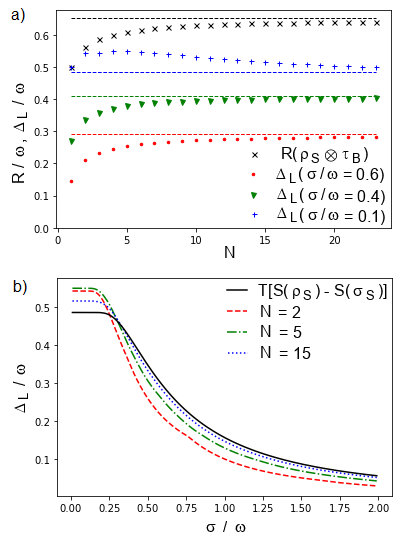}
    \caption{\emph{Optimal work and locked energy.} (a) Graph presents how the optimal work given by the ergotropy $R(\hat \rho_S \otimes \hat \tau_B^{(N)})$ and locked energy $\Delta_L(\hat \rho_S, \hat \rho_W, \hat \tau_B^{(N)})$ depends on the number of qubits in the heat bath $N$ for different values of the scaled standard deviation $\sigma/\omega$ of the Gaussian state of the weight. Horizontal lines correspond 
   to the thermodynamic limits (with $N \to \infty$) given by the free energy $F(\hat \rho_S)-F(\hat \tau_S)$ and entropy $T[S(\hat \rho_S) - S(\hat \sigma_S)]$ differences. (b) The vanishing of the locked energy for different size of the heat bath $N$ with respect to the parameter $\sigma/ \omega$.}
    \label{fig:free_energy_vs_ergotropy}
\end{figure}

%\emph{Example.}  
\subsection*{Example}
Let us now consider a particular example in order to illustrate how the finite-size bath and state of the weight affect the work extraction process. We would like to concentrate on a system $S$ given by the qubit in a coherent `plus state', i.e. $\hat \rho_S = \dyad{+}_S$, where $\ket{+}_S = \frac{1}{\sqrt{2}} (\ket{0}_S + \ket{1}_S)$, and with Hamiltonian $\hat H_S = \omega \dyad{1}_S$. Next, as a model of a bath we take a collection of qubits with different energy gaps, namely the bath Hamiltonian is given by:
\begin{equation} \label{bath_hamiltonian}
    \hat H_B^{(N)} = \bigotimes_{k=1}^N \omega_k \dyad{1_k}_B
\end{equation}
where $\omega_k = T \log[\frac{1-k\delta}{k\delta}]$ and $\delta = \mathcal{Z}^{-1}_S e^{-\beta \omega}/N$. The choice of the heat bath is dictated by its property that in the limit of infinite number of qubits a saturation of inequality \eqref{ergotropy_vs_free_energy} is achieved \cite{Skrzypczyk2014}. Finally, we take the weight in a pure state given by a Gaussian superposition of energy states, i.e. $\hat \rho_W = \dyad{\psi}_W$ such that
\begin{equation} \label{weight_state}
    \ket{\psi}_W = (2\pi \sigma^2)^{-1/4} \int d\varepsilon \ e^{-\frac{\epsilon^2}{4\sigma^2}} \ket{\varepsilon}_W,
\end{equation}
where the vector is solely parameterized by the standard deviation $\sigma$.

Within this model we numerically calculate the optimal work 
\begin{equation}
    W_{max} = R \left(\dyad{+}_S \otimes \hat \tau_B^{(N)} \right)
\end{equation}
for baths $\hat \tau_B^{(N)}$ with different number of qubits $N$ \eqref{bath_hamiltonian}. The results are presented in the Fig. 1(a). In particular, in the graph we plot the ultimate bound given by the difference of free energy $F(\hat \rho_S) - F(\hat \tau_S)$ (see inequality \eqref{ergotropy_vs_free_energy}) and reveal how the optimal work $W_{max}$ converges to this limit with increasing number of qubits. We emphasize that ergotropy $R (\hat \rho_S \otimes \hat \tau_B^{(N)})$ is an increasing function with respect to the growing size of the heat bath $N$.

Next, we calculate the locked energy, namely
% \begin{equation}
% \begin{split}
%     &\Delta_L (\hat \rho_S, \hat \rho_W, \hat \tau_B) = \\
%     &= R \Big(\dyad{+}_S \otimes \hat \tau_B^{(N)} \Big) - R \Big(\hat \xi_S \otimes \hat \tau_B^{(N)} \Big) - \frac{\omega \gamma}{2},
% \end{split}
% \end{equation}
\begin{equation}
    \Delta_L = R \left(\dyad{+}_S \otimes \hat \tau_B^{(N)} \right) - R \left(\hat \xi_S \otimes \hat \tau_B^{(N)} \right) - \frac{\omega \gamma}{2},
\end{equation}
where $\gamma = \exp[-\frac{\omega^2}{8\sigma^2}]$ and %$\hat \xi_S = \frac{1}{2}(1+\gamma) \dyad{0}_S + \frac{1}{2} (1-\gamma) \dyad{1}_S$.
\begin{equation}
\hat \xi_S = \frac{1}{2}(1+\gamma) \dyad{0}_S + \frac{1}{2} (1-\gamma) \dyad{1}_S.
\end{equation}
It is seen that the Gaussian model of the weight affects the locked energy only via a single parameter $\sigma / \omega$, i.e. a ratio between a standard deviation of the work reservoir wave packet $\sigma$ and the energy gap of a qubit $\omega$. In analogy to the optimal work, in the Fig. 1(a) we analyze how the locked energy depends on the number of qubits in the bath $N$ and compare it to the thermodynamic limit given by $T[S(\hat \rho_S) - S(\hat \sigma_S)]$ \eqref{locked_ergotropy_bound}. An intriguing observation is that as long as the ergotropy is an increasing function with respect to the growing size of the heat bath, the locked energy is not. It is observed that for some values of $\sigma / \omega$ it can be non-monotonic with respect to the number of qubits $N$, i.e. adding a qubit can increase likewise decrease the locked energy.   

Further, in the Fig 1(b) we present how quickly the locked energy vanishes with increasing value of the ratio $\sigma / \omega$. Two interesting features are observed here. First for high values of $\sigma / \omega$ the locked energy is increasing with the size of the bath $N$, however, this order is changed for low values and becomes non-monotonic. Secondly, for low values we observe a \emph{plateau}, i.e. the locked energy almost stay constant with growing width of the weight wave packet. Notice that in the limit $\sigma \to 0$ state of the weight tends to the energy state with the maximal locked energy and in the limit $\sigma \to \infty$ it tends to the time state for which the locked energy vanishes.  

%\emph{Conclusions.} 
\subsection*{Conclusions}
We recognized ergotropy as a proper resource regarding the work extraction process with an explicit energy storage given by the translationally invariant weight. The ergotropy on its own can be defined as the optimal work extracted from closed systems driven by the time-dependent and cyclic Hamiltonians, what proofs an important connection between those two frameworks. Nevertheless, we stress that there is no full equivalence between them, since models with an implicit energy storage do not involve the concept of the \emph{locked energy}, i.e. the part from coherences that contribute to ergotropy (or free energy) but cannot be extracted as a work. Indeed, one of the main difference between classical and quantum thermodynamics is that quantum systems are able to perform the work via coherences. However, here we reveal that it is only possible if the work reservoir has coherences as well and the locked energy naturally emerges if we treat it explicitly. In other words, for the quantum process of work extraction it is crucial that the weight is the energy reservoir and the reservoir of coherences likewise. Consequently, we provide a quantitative definition of the \emph{ideal work reservoir}, i.e. the energy storage that is able to the full work extraction from coherences, what really is the case in non-autonomous approach. 

Furthermore, we analyze the ergotropy of the non-equilibrium quantum system in a contact with arbitrary finite-size heat bath. In the light of the resource theory, such a Gibbs state of the bath is treated as a costless, i.e. it can be for free attached to, and discarded from the system. Due to the non-additivity of the ergotropy, the state of the heat bath activates the non-equilibrium state of the system and consequently both of them form the entire thermodynamic resource, given by the total ergotropy. This can be simply interpreted as a maximal work that can be extracted from such a quantum state. 
Moreover, one of the most important result of this work is an establishment of the general relation between the ergotropy and free energy for systems coupled to the heat bath, what provides a bridge between microscopic and macroscopic thermodynamics. We show that the total ergotropy of the quantum system and finite-size heat bath is indeed the generalization of the non-equilibrium free energy, and it converges to the latter in the thermodynamic limit. %Further, this naturally leads to the definition of the \emph{generic heat bath}, i.e. the one which provides an ability of the optimal work extraction. 

Finally, relation between ergotropy and free energy leads us to the thermodynamic limit of the locked energy. This provides an interesting formula, expressed in terms of the von Neumann entropy, that from one side is fully quantum, since refers to the extraction of work from coherences (i.e. requires the coherent state of the system and the energy storage likewise), however, on the other side involves the classical notion of the macroscopic heat bath. 

\begin{acknowledgements}
The author thanks Micha{\l} Horodecki, Pawe{\l} Mazurek,  Tony  Short and Patryk Lipka-Bartosik for helpful and inspiring  discussions.  This research was supported by the National Science Centre, Poland, through grant  SONATINA 2 2018/28/C/ST2/00364.
\end{acknowledgements}

\newpage

\bibliography{bib}

\onecolumngrid
\newpage
\appendix

\section{}
\subsection{Average energy, passive energy and ergotropy}
Let us consider a system $\hat \rho$ with free Hamiltonian $\hat H$. We define the following quantities:
\begin{equation}
    E(\hat \rho) = \Tr[\hat H \hat \rho], \  P(\hat \rho) = \min_{\hat V} \Tr[ \hat V^\dag \hat H \hat V \hat \rho], \ R(\hat \rho) = E(\hat \rho) - P(\hat \rho), 
\end{equation}
i.e. average energy $E(\hat \rho)$, passive energy $P(\hat \rho)$ and ergotropy $R(\hat \rho)$, 
% - passive energy:
% \begin{equation}
%     P(\hat \rho) = \min_{\hat V} \Tr[ \hat V \hat H \hat V^\dag \hat \rho]
% \end{equation}
% - ergotropy:
% \begin{equation}
%     R(\hat \rho) = E(\hat \rho) - P(\hat \rho) 
% \end{equation}
where $\hat V$ is arbitrary unitary acting on the system Hilbert space. 
\subsection{Energy and time states of the weight}
We assume that the energy storage is given by the weight $\mathcal{W}$. In particular, this implies that its energy spectrum is continuous, namely
\begin{equation}
    \hat H_W = \int d \varepsilon \ \varepsilon \dyad{\varepsilon}_W.
\end{equation}
We further define time states $\ket{t}_W$ as the canonically conjugates states with respect to the energy states $\ket {\varepsilon}_W$, i.e. transition from one to another basis is given by the Fourier transform:
\begin{equation}
    \ket{t}_W = \int d \varepsilon \ e^{i \varepsilon t} \ket{\varepsilon}_W. 
\end{equation}
Next, we consider a shift operator which is defined by an action on energy eigenstates: $\hat \Gamma_\delta \ket{\varepsilon}_W = \ket{\varepsilon + \delta}_W$. According to this relation it follows that:
\begin{equation}
   \hat \Gamma_\delta \ket{t}_W = \int d \varepsilon \ e^{i \varepsilon t} \ \hat \Gamma_\delta \ket{\varepsilon}_W = \int d \varepsilon \ e^{i \varepsilon t} \ket{\varepsilon + \delta}_W = e^{-i \delta t} \int d \varepsilon \ e^{i \varepsilon t} \ket{\varepsilon}_W = e^{-i \delta t} \ket{t}_W,
\end{equation}
i.e. time states are eigenstates of the shift operator. 
\subsection{Energy-conserving and translationally invariant unitary}
We consider a quantum system $\mathcal{S}$ coupled to a heat bath $\mathcal{B}$, prepared in a Gibbs state $\hat \tau_B$, and the weight $\mathcal{W}$. Then, we investigate a unitary protocol  
$\hat \rho_{SW} \otimes \hat \tau_B \to \hat U \hat \rho_{SW} \otimes \hat \tau_B \hat U^\dag$, such that the evolution operator is the energy-conserving translationally invariant unitary, i.e. it satisfies the following commutation relations: 
\begin{equation} \label{commutation_relations}
   [\hat U, \hat H_S + \hat H_B + \hat H_W] = 0, \  [\hat U, \hat \Gamma_{\epsilon}] = 0
\end{equation}
where $\hat H_k$ is a free Hamiltonian of $k=S,B,W$ subsystem and $\hat \Gamma_\epsilon$ is a weight shift operator (for arbitrary real $\epsilon$). 

In the paper \cite{Lobejko2020} it was proven that unitary $\hat U$ obeying conditions \eqref{commutation_relations} can be always written in the form:
\begin{equation} \label{unitary_form}
    \hat U = \hat S^\dag (\hat V_{SB} \otimes \mathbb{1}_W) \hat S,
\end{equation}
where $\hat V_{SB}$ is some unitary acting on the system and bath Hilbert space and $\mathbb{1}_W$ is the identity operator acting on the weight, and $\hat S$ is a kind of \emph{control-shift operator} defined as follows:
\begin{equation}
    \hat S = \sum_{i,j} \dyad{\epsilon_i}_{S} \otimes \dyad{\varepsilon_j}_{B} \otimes \hat \Gamma_{\epsilon_i+\varepsilon_j}
\end{equation}
where $\ket{\epsilon_i}_{S}$ is an eigenstate of the system Hamiltonian $\hat H_S$ and $\ket{\varepsilon_j}_{B}$ is an eigenstate of the $\hat H_B$. 

\subsection{Work and control-marginal state}
From the Eq. \eqref{unitary_form} further follows that work is equal to:
\begin{equation} \label{work_appendix}
    W = \Tr[\hat H_W (\hat U \hat \rho_{SBW} \hat U^\dag - \hat \rho_{SBW})] =  \Tr[\hat H_{SB} (\hat \sigma_{SB} - \hat V_{SB} \hat \sigma_{SB} \hat V_{SB}^\dag)],
\end{equation}
where $\hat \sigma_{SB} = \Tr_W[\hat S \hat \rho_{SBW} \hat S^\dag]$ is the so-called \emph{control-marginal state}. For the product state $\hat \rho_{SBW} = \hat \rho_S \otimes \hat \tau_B \otimes \hat \rho_W$, we have
\begin{equation}
\begin{split}
    \hat \sigma_{SB} &= \Tr_W [ \hat \Gamma_{\epsilon_i + \varepsilon_k} \hat \rho_W \hat \Gamma_{\epsilon_j + \varepsilon_l}^\dag] \sum_{i,j,k,l} \dyad{\epsilon_i}_{S} \hat \rho_S \dyad{\epsilon_j}_{S} \otimes \dyad{\varepsilon_k}_{B} \hat \tau_B \dyad{\varepsilon_l}_{B}  \\
    &= \Tr_W [ \hat \Gamma_{\epsilon_i} \hat \rho_W \hat \Gamma_{\epsilon_j}^\dag] \sum_{i,j} \dyad{\epsilon_i}_{S} \hat \rho_S \dyad{\epsilon_j}_{S} \otimes \hat \tau_B \equiv \hat \sigma_S \otimes \hat \tau_B.
\end{split}
\end{equation}

Now, we would like to derive an alternative form for the control-marginal state $\hat \sigma_S$ (for the product states). Let us represent the density matrix of the weight in a time states basis, i.e. $\hat \rho_W = \int dt \ ds \ \dyad{t}_W \hat \rho_W \dyad{s}_W$. Putting it into the above formula we obtain:
\begin{equation} \label{sigma_appendix}
\begin{split}
    \hat \sigma_S &= \int dt \ ds \ \Tr_W [ \hat \Gamma_{\epsilon_i} \dyad{t}_W \hat \rho_W \dyad{s}_W \hat \Gamma_{\epsilon_j}^\dag] \sum_{i,j} \dyad{\epsilon_i}_{S} \hat \rho_S \dyad{\epsilon_j}_{S} \\
    &= \int dt \ ds \ \Tr_W [ e^{-i\epsilon_i t} \dyad{t}_W \hat \rho_W  \dyad{s}_W e^{i\epsilon_j s}] \sum_{i,j} \dyad{\epsilon_i}_{S} \hat \rho_S \dyad{\epsilon_j}_{S} = \int dt \ p(t) \ e^{-i \hat H_S t} \hat \rho_S  e^{i \hat H_S t}
\end{split}
\end{equation}
where $p(t) = \Tr_W[\hat \rho_W \dyad{t}_W]$. 

\subsection{Optimal work extraction}
\begin{theorem}
For arbitrary transition $\hat \rho_S \to \hat \rho_S' = \Tr_{BW} [\hat U \hat \rho_{SW} \otimes \hat \tau_B \hat U^\dag]$ the work extracted by the weight is equal to:
\begin{equation}
    W = R(\hat \sigma_S \otimes \hat \tau_B) - R(\hat V_{SB} \hat \sigma_S \otimes \hat \tau_B \hat V_{SB}^\dag).
\end{equation}
Moreover, there exist a unitary $\hat V_{SB}$ such that $R(\hat V_{SB} \hat \sigma_S \otimes \hat \tau_B \hat V_{SB}^\dag) = 0$, and then the optimal work is given by:
\begin{equation} \label{optimal_tight}
    W_{max} = R(\hat \sigma_S \otimes \hat \tau_B). 
\end{equation}
\end{theorem}
\begin{proof}
According to Eq. \eqref{work_appendix} we have
\begin{equation}
\begin{split}
    W &= \Tr[\hat H_{SB} (\hat \sigma_{SB} - \hat V_{SB} \hat \sigma_{SB} \hat V_{SB}^\dag)] - \max_{\hat U - \text{unitary}} \Tr[\hat U \hat \sigma_{SB} \hat U^\dag] + \max_{\hat U - \text{unitary}} \Tr[\hat U \hat \sigma_{SB} \hat U^\dag] \\
    &= \max_{\hat U - \text{unitary}} \Tr[\hat H_{SB} (\hat \sigma_{SB} - \hat U \hat \sigma_{SB} \hat U^\dag)] + \max_{\hat U - \text{unitary}} \Tr[\hat H_{SB} (\hat U \hat \sigma_{SB} \hat U^\dag - \hat V_{SB} \hat \sigma_{SB} \hat V_{SB}^\dag)] \\
    &= R(\hat \sigma_{SB}) - \max_{\hat U - \text{unitary}} \Tr[\hat H_{SB} (\hat V_{SB} \hat \sigma_{SB} \hat V_{SB}^\dag - \hat U \hat V_{SB} \hat \sigma_{SB} \hat V_{SB}^\dag \hat U^\dag)] = R(\hat \sigma_{SB}) - R(\hat V_{SB} \hat \sigma_{SB} \hat V_{SB}^\dag)
\end{split}
\end{equation}
The second part follows from the fact that the arbitrary state $\hat \sigma_{SB}$ can be unitarly transformed to the passive state (i.e. with zero ergotropy).
\end{proof}

\subsection{Ergotropy and free energy}
\begin{theorem}
Let $\hat \rho_S$ and $\hat \xi_S$ are arbitrary quantum states and $\hat \tau_B$ is the Gibbs state. Further, $\hat \rho_p$ and $\hat \xi_p$ are passive states (i.e. with minimal energy) obtained through the unitary channel from $\hat \rho_S \otimes \hat \tau_B$ and $\hat \xi_S \otimes \hat \tau_B$, respectively. Then, if $F(\hat \xi_p) \le F(\hat \rho_p)$ it implies that
\begin{equation} \label{ergotropy_vs_free_energy_diff_apen}
    R(\hat \rho_S \otimes \hat \tau_B) - R(\hat \xi_S \otimes \hat \tau_B) \le F(\hat \rho_S) - F(\hat \xi_S).
\end{equation}
\end{theorem}
\begin{proof}
From the definition of the free energy and assumption $F(\hat \xi_p) \le F(\hat \rho_p)$ we obtain:
\begin{equation}
    E(\hat \xi_p) - E(\hat \rho_p) - T [S(\hat \xi_p) - S(\hat \rho_p)] \le 0
\end{equation}
Further, we have $S(\hat \rho_p) = S(\hat \rho_S \otimes \hat \tau_B)$ and $S(\hat \xi_p) = S(\hat \xi_S \otimes \hat \tau_B)$ such that the above inequality can be rewritten in the form:
\begin{equation}
 E(\hat \sigma_p) - E(\hat \xi_S \otimes \hat \tau_B) - E(\hat \rho_p) + E(\hat \rho_S \otimes \hat \tau_B) \le F(\hat \rho_S \otimes \hat \tau_B) - F(\hat \xi_S \otimes \hat \tau_B).
\end{equation}
Finally, since $ F(\hat \rho_S \otimes \hat \tau_B) - F(\hat \xi_S \otimes \hat \tau_B) =  F(\hat \rho_S) - F(\hat \xi_S)$ and $R(\hat \rho_S \otimes \hat \tau_B) = E(\hat \rho_S \otimes \hat \tau_B) - E(\hat \rho_p)$ (and the same for $\hat \xi_S$) we obtain inequality \eqref{ergotropy_vs_free_energy_diff_apen}. 
\end{proof}

\begin{corollary}
For arbitrary state $\hat \rho_S$ and arbitrary Gibbs state $\hat \tau_B$ it is valid:
\begin{equation}
    R(\hat \rho_S \otimes \hat \tau_B) \le F(\hat \rho_S) - F(\hat \tau_S).
\end{equation}
\end{corollary}
\begin{proof}
Let us take the state $\hat \xi_S$ as a Gibbs state in the same temperature as $\hat \tau_B$, i.e. $\hat \xi_S = \hat \tau_S$, then $\hat \xi_p = \hat \tau_S \otimes \hat \tau_B = \hat \xi_S \otimes \hat \tau_B$. Moreover, for arbitrary state $\hat \rho_S$ it is satisfied an inequality $F(\hat \rho_p) \ge F(\hat \xi_p)$, since $\hat \xi_p$ is a Gibbs state for which free energy has minimum. Finally, since $R(\hat \tau_S \otimes \hat \tau_B) = 0$ and from inequality \eqref{ergotropy_vs_free_energy} follows what was to be shown. 
\end{proof}

\subsection{Locked energy}

\begin{proposition}
For the quantum state $\hat \rho_S$ and its control-marginal state $\hat \sigma_S$ \eqref{sigma_appendix} we have $E(\hat \sigma_S \otimes \hat \tau_B) = E(\hat \rho_S \otimes \hat \tau_B)$ and $P(\hat \sigma_S \otimes \hat \tau_B) \ge P(\hat \rho_S \otimes \hat \tau_B)$ from which follows that 
\begin{equation}
    \Delta_L(\hat \rho_S, \hat \rho_W, \hat \tau_B) = R(\hat \rho_S \otimes \hat \tau_B) - R(\hat \sigma_S \otimes \hat \tau_B) \ge 0.
\end{equation}
\end{proposition}
\begin{proof}
Firstly, let us show that the average energy of the control-marginal state is equal to the marginal state, namely
\begin{equation} \label{energy_equality}
\begin{split}
        &E(\hat \sigma_S \otimes \hat \tau_B) = \Tr[(\hat H_S + \hat H_B) \hat \sigma_S \otimes \hat \tau_B] = \int dt \ p(t) \Tr[(\hat H_S + \hat H_B) \hat U_t \hat \rho_S \hat U_t^\dag \otimes \hat \tau_B] =  \\
        &=  \int dt \ p(t) \Tr[(\hat U_t^\dag \hat H_S \hat U_t+ \hat H_B)  \hat \rho_S \otimes \hat \tau_B] = \Tr[(\hat H_S + \hat H_B)  \hat \rho_S \otimes \hat \tau_B] = E(\hat \rho_S \otimes \hat \tau_B).
\end{split}
\end{equation}
Secondly, the passive energy obeys the following inequality:
\begin{equation}
\begin{split}
        &P(\hat \sigma_S \otimes \hat \tau_B) = \min_{\hat V_{SB}} \Tr[\hat V_{SB}^\dag (\hat H_S + \hat H_B) \hat V_{SB} \hat \sigma_S \otimes \hat \tau_B] = \min_{\hat V_{SB}} \int dt \ p(t) \Tr[\hat U_t^\dag \hat V_{SB}^\dag (\hat H_S + \hat H_B) \hat V_{SB} \hat U_t \hat \rho_S \otimes \hat \tau_B] \\
        &\ge \int dt \ p(t) \min_{\hat V_{SB}} \Tr[\hat U_t^\dag \hat V_{SB}^\dag (\hat H_S + \hat H_B) \hat V_{SB} \hat U_t \hat \rho_S \otimes \hat \tau_B] = \int dt \ p(t) P(\hat \rho_S \otimes \hat \tau_B) = P(\hat \rho_S \otimes \hat \tau_B).
\end{split}
\end{equation}
Finally, we obtain
\begin{equation}
    R(\hat \sigma_S \otimes \hat \tau_B) = E(\hat \sigma_S \otimes \hat \tau_B) - P(\hat \sigma_S \otimes \hat \tau_B) \le E(\hat \rho_S \otimes \hat \tau_B) - P(\hat \rho_S \otimes \hat \tau_B) = R(\hat \rho_S \otimes \hat \tau_B).
\end{equation}
\end{proof}

\begin{corollary}
If for a Gibbs state $\hat \tau_B$ holds an equality $R(\hat \rho_S \otimes \hat \tau_B) = F(\hat \rho_S) - F(\hat \tau_S)$ for arbitrary state $\hat \rho_S$, then locked energy is equal to:
\begin{equation}
    \Delta_L(\hat \rho_S, \hat \rho_W, \hat \tau_B) = R(\hat \rho_S \otimes \hat \tau_B) - R(\hat \sigma_S \otimes \hat \tau_B) = T[S(\hat \rho_S) - S(\hat \sigma_S)].
\end{equation}
\end{corollary}

\end{document}